\documentclass[envcountsect,runningheads,a4paper]{llncs}
\usepackage{graphicx}
\usepackage{algorithm}
\usepackage{algpseudocode}

\algnewcommand{\LineComment}[1]{\State \(\triangleright\) #1}
\usepackage{xcolor}
\usepackage{tikz-cd}
\usepackage{soul} 
\usepackage{amssymb,amsmath}
\usepackage{comment}

\def\Fq{\mathbb{F}_q}

\def\Fqn{\mathbb{F}_{q^n}}

\newcommand{\D}[1][D]{\ensuremath{\mathcal{#1}}}

\begin{document}
\title{Introducing locality in some generalized AG codes}

%
%
\author{
Bastien Pacifico \inst{1}
}
%
\institute{LIRMM, Université de Montpellier\\Montpellier, France \\
\email{bastien.pacifico@gmail.com}
}
\maketitle              

\begin{abstract}
In 1999, Xing, Niederreiter and Lam introduced a generalization of AG codes using the evaluation at non-rational places of a function field. In this paper, we show that one can obtain a locality parameter $r$ in such codes by using only non-rational places of degrees at most $r$. This is, up to the author's knowledge, a new way to construct locally recoverable codes (LRCs). We give an example of such a code reaching the Singleton-like bound for LRCs, and show the parameters obtained for some longer codes over $\mathbb F_3$. 
We then investigate similarities with certain concatenated codes. Contrary to previous methods, our construction allows one to obtain directly codes whose dimension is not a multiple of the locality. Finally, we give an asymptotic study using the Garcia-Stichtenoth tower of function fields, for both our construction and a construction of concatenated codes. We give explicit infinite families of LRCs with locality 2 over any finite field of cardinality greater than 3 following our new approach.
\end{abstract}

\section{Introduction}

\subsection{Known results}

Locally Recoverable Codes (LRCs) are a popular topic lately, in particular for their potential applications in distributed storage \cite{local}.  The locality consists in the possibility of recovering one corrupted symbol using a small amount of other symbols. 
More precisely, a code is said to have locality $r$ if any symbol of a codeword can be obtained using at most $r$ other symbols \cite{pyram,local}.
It was proven in \cite{local} that a $[n,k,d]$ linear code with locality $r$ verifies 
\begin{equation}\label{SingletonLRC}
    d\leq n-k-\left\lceil\frac{k}{r}\right\rceil+2.
\end{equation}
A LRC is said to be optimal when the equality is reached in this bound. 
There exist several constructions of optimal 
LRCs. The first ones were given in \cite{pyram,rank,tapadi}, but they required to use an alphabet of exponential size compared to the code length.
The construction of Tamo and Barg \cite{taba} provides optimal codes of length upper bounded by the size of the alphabet and moreover with constraint on the locality due to the existence of good polynomials.
There also exist codes that reach the bound (\ref{SingletonLRC}) of length greater than size of the alphabet. Such construction can be obtained using for instance algebraic surfaces \cite{bama,salgadovoloch}. In fact, it has been proven in \cite{guxiyu} that the length of an optimal LRC is at most $\mathcal{O}(dq^3)$, where $q$ is the size of the alphabet and for minimum distance $d\geq 5$.
On the other side, a classical approach in coding theory is to study the best that can be obtained for a fixed alphabet size. Fundamental works on this topic were done in \cite{cama,tabafr}. They gave tight bounds and achievability results, such as a Gilbert-Varshamov bound for LRCs. Several constructions of families of codes were then given, for instance or for instance using concatenated codes \cite{cama}, or in \cite{tabavl,maxi,luca,tanal,xingalllast}.  Other considerations on the topic are the correction of multiple erasures and the correction from multiple recovery sets. Details can be found for instance in \cite{taba,tabavl}.

In this paper, we consider the generalized AG-codes introduced in \cite{xinggene} by Xing, Niederreiter and Lam. The well known AG-codes defined by Goppa in \cite{goppa}
are given by the evaluation at rational places (i.e. of degree 1) of functions of an algebraic function field defined over $\Fq$. The generalization of Xing et al. consists in using not only the evaluation at rational places, but at places of higher degrees. In fact, the evaluation at a place of degree $d$ is an element in the residue class field, that is isomorphic to $\mathbb{F}_{q^d}$. It follows that a codeword composed by some evaluations at places of different degree would be polyalphabetic. To address this difficulty, the solution proposed in \cite{xinggene} is to encode the evaluation at non rational places with a $\Fq-$linear code. 

\subsection{Contributions and organization}

The observation behind this document is as follows: if we apply the construction of \cite{xinggene} using non-rational places of degrees at most $r>1$, we can obtain linear codes with locality $r$. It turns out that some of these codes have good or optimal parameters with respect to the Singleton bound for LRCs (\ref{SingletonLRC}). There are similarities between our codes and some concatenated codes, especially to those introduced by Cadambe and Mazumbar in \cite[Section VI. A.]{cama}. 
In order to make a comparison and investigate their differences, we give 
a construction of LRCs obtained by concatenation using an AG-code as outer code. More precisely, a construction from \cite{cama} uses a RS code  as the outer code. We consider a similar construction by using an AG outer code and show its parameters are similar to those of our new construction. Using the recursively defined tower of function fields of Garcia and Stichtenoth \cite{gast}, we give an asymptotic study of both our new construction using generalized AG-codes and the construction of concatenated codes. An important difference is that our new approach from generalized AG-codes allows to construct directly codes whose dimension is not a multiple of the locality, contrary the the one using concatenated codes or the best-known constructions (e.g. \cite{tabavl}).

The paper is organized as follows. In Section 2, we recall the basics of LRCs and concatenated codes. In Section 3, we give the definitions and results of function field theory that we shall use, and present the generalization of AG codes of \cite{xinggene}. In Section 4, we explain how one can obtain locality in these codes, and give an optimal example. In Section 4, we give two explicit families of LRCs, one using concatenated codes and the second with our new approach. We give an asymptotic study. In particular, we show the existence of an infinite family of LRC with locality 2 over finite fields of cardinality greater than 3 and give their parameters.

\section{Locally Recoverable Codes (LRCs)}

\subsection{Generalities}

In what follows, we denote by $[n,k,d]$ a linear code over $\Fq$ with length $n$, dimension $k$ and minimum distance $d$. Such a code is MDS if its parameters reach the equality in the Singleton bound $d\leq n-k+1$. The codes of Reed-Solomon (RS) are a well known example of MDS codes. More precisely, a Reed-Solomon code $\mathrm{RS}(n,k)$ of length $n$ and dimension $k$ is defined by the image of an application
$$\mathrm{RS}_{n,k}:\begin{array}{lcl}
  \Fq[x]_{<k}  & \longrightarrow & \Fq^n \\
   f & \longmapsto & \left(f(\alpha_1),\ldots,f(\alpha_n)\right),
\end{array}$$
where $\alpha_1,\ldots,\alpha_n$ are distinct elements of $\Fq$.
Throughout this paper, we focus on the notion of locally recoverable codes (LRCs).

\begin{definition}
Let $\mathcal{C}\subset \Fq^n$ be a $\Fq-$linear code. The code $\mathcal{C}$ is locally recoverable with locality $r$ if every symbol of a codeword $c=(c_1,\ldots,c_n)\in\mathcal{C}$ can be recovered using a subset of at most $r$ other symbols. The smallest such $r$ is called the locality of the code. 
\end{definition}
This means that for each $1\leq i \leq n$, there exists a recovery set $N_i\subset\{1,\ldots,n\}\smallsetminus\{i\}$ of cardinality at most $r$ such that $c_i$ can be obtained from $\{c_j\}_{j\in N_i}$. In particular, we seek for codes with a small locality $r$.
There exists a Singleton-like bound for LRCs \cite{local} and an upper bound for the rate of codes with locality $r$, given in \cite[Theorem 2.1]{taba}, that we recall.
\begin{theorem}
Let $\mathcal{C}$ be a $q-$ary linear code with parameters $[n,k,d]$ with locality $r$. 
The rate of $\mathcal{C}$ verifies 
$$\frac{k}{n}\leq \frac{r}{r+1}.$$
The minimum distance $d$ of $\mathcal{C}$ verifies
$$d\leq n-k-\left\lceil\frac{k}{r}\right\rceil+2.$$
\end{theorem}

\begin{example}
The code $\mathrm{RS}(n,k)$ is a trivial example of LRC with locality $k$, since every symbol can be reconstructed by using $k$ others. It reaches the Singleton-like bound since $n-k-\lceil\frac{k}{k}\rceil+2=n-k+1$. However, it is a bad example of LRC since the locality is equal to the dimension. 
\end{example}

However, the previous results do not take into account the size of the alphabet, i.e. the cardinality of the base field. There are several bounds considering this constraint. In \cite{cama} and \cite[Theorem 5.1]{tabafr}, the authors gave a Gilbert-Varshamov-type bound for LRC (\ref{GV}), where $R_q(r,\delta)$ denotes the
asymptotic bound on the rate of $q-$ary locally repairable codes with locality $r$ and relative minimum distance $\delta$.

\begin{equation}\label{GV}
\begin{aligned}
    R_q(r,\delta)\geq &1-\min_{0\leq s \leq 1}\left[\frac{1}{r+1}\log_q\left((1+(q-1)s)^{r+1}\right.\right.
    \\&\left.\left.+(q-1)(1-s)^{r+1}\right)-\delta\log_q s\right]. 
\end{aligned}
\end{equation}

The construction of Tamo-Barg-Vladuts \cite{tabavl} is known to exceed this bound. This result is also obtained by \cite{maxi,xingalllast}.
Some achievability results considering a fixed alphabet size have been obtained via concatenated codes \cite{cama,xingalllast}. 

\subsection{Concatenated codes}

Concatenated codes were introduced by Forney \cite{forney} in 1965. This name comes from the idea of successively applying two encoders. It consists in first using an outer code over a large alphabet, then use an inner code to encode the codeword symbols of the first code. In our framework, a concatenated code can be defined as follow. 

\begin{definition}
    Let $\mathcal C_\mathrm{out}$ be a $q^{k'}-ary$ linear code of parameters $[n,k,d]$ and $\mathcal C_\mathrm{in}$ be a $q-ary$ linear code of parameters $[n',k',d']$ such that $$
    \mathcal C_\mathrm{out}(m)=(c_1,,\ldots,c_n),$$ where $m\in\mathbb F_{q^{k'}}^k$ and $c_1,\ldots,c_n\in\mathbb F_{q^{k'}}$.
    Then the concatenated code $\mathcal C_\mathrm{conc}$ of $\mathcal{C}_\mathrm{in}$ and $\mathcal{C}_\mathrm{out}$ is defined by 
    $$\mathcal{C}_\mathrm{conc}(m)=\left(\mathcal{C}_\mathrm{in}(c_1)~\mid~\cdots~\mid~\mathcal{C}_\mathrm{in}(c_n)\right).$$
\end{definition}

Note that the locality of a concatenated code is given by the one of the inner code \cite[Theorem 4.1]{xingalllast}.
Recall also that such code verifies the following properties.

\begin{proposition}\label{concprop}
The code $\mathcal{C}_\mathrm{conc}$ is a $[nn',kk',\geq dd']$ linear code over $\mathbb F_q$.
\end{proposition}

In \cite[Theorem 2]{cama}, the authors used concatenated codes to obtain asymptotic achievability results on binary LRCs. More precisely, they used an outer random $q^r-ary$ linear code and the $q-ary$ single parity check code of length $r+1$ as the inner code.
In \cite{xingalllast}, the authors also used concatenated codes to obtain some dimension-optimal locally repairable codes. Moreover, they also used some shortening techniques to obtain dimension-optimal LRCs whose dimension is not a multiple of the locality. 

\section{Generalized AG-Codes}

\subsection{Algebraic function fields}
Let $\Fq$ be the field with $q$ elements, and let $F/\Fq$ be an algebraic function field of genus $g=g(F)$ over $\Fq$. For $\mathcal{O}$ a valuation ring, a place $P$ is defined to be $\mathcal{O}\smallsetminus\mathcal{O}^\times$. The evaluation of a function at $P$ is an element of the residue class field $F_P$, that is isomorphic to $\mathbb F_{q^d}$, $d$ being the degree of the place. A rational place is a place of degree 1. A divisor $\D$ is defined as a formal sum of places, and we denote by $Supp(\D)$ the support of $\D$ and $\mathcal{L}(\D)$ the corresponding Riemann-Roch space. Details about algebraic function fields can be found in 
\cite{stic2}.

In the following, obtaining infinite families of codes with our construction relies on the existence of families of function fields with a large number of places of given degree.
In this context, we consider infinite sequences of algebraic function fields.

\begin{definition}
    An infinite sequence of function field over $\Fq$ is a sequence of $\mathcal{F}=(F_1,\ldots,F_\ell,\ldots)$ of function fields $F_\ell/\Fq$ such that for all $\ell$, we have $g(F_{\ell+1})>g(F_\ell)$.
\end{definition}

\noindent Moreover, let us introduce the Drinfeld-Vladut Bound at order $r$.

\begin{definition}[Drinfeld-Vladut Bound of order $r$]
Let $F/\Fq$ be a function field over $\Fq$ and $B_r(F/\Fq)$ denotes its number of places of degree $r$. Let
$$B_r(q,g)=\max\{B_r(F/\Fq~\mid~F/\Fq)\text{ is a function field over }\Fq\text{ of genus g}\}.$$
Then,
$$\limsup_{g\longrightarrow+\infty}\frac{B_r(q,g)}{g}\leq\frac{1}{r}(q^\frac{r}{2}-1).$$
\end{definition}
For $r=1$, this gives the usual Drinfeld-Vladut Bound on the number of rational places. There exist several families of function field reaching this bound, such as the Garcia-Stichtenoth tower of function fields \cite{gast}. Such towers are recalled and used in Section \ref{asympsection}.

\subsection{Generalized AG codes}
Let $F/\Fq$ be an algebraic function field defined over $\Fq$ of genus $g$.
An AG code is defined by the evaluation at rational places of functions from a Riemann-Roch space. More precisely, considering a function field $F/\Fq$, let $G$ be a divisor and $\D=P_1+\cdots+P_n$, where the $P_i$ are rational places. The AG code $C(\D,G)$ is defined by the map

$$C(\D,G):\begin{array}{lcl}
  \mathcal{L}(G)  & \longrightarrow & \Fq^n \\
   f & \longmapsto & \left(f(P_1),\ldots,f(P_n)\right).
\end{array}$$

In \cite{xinggene}, the authors extended this construction to the use of non-rational places. Following their work, we use the notations:
\begin{itemize}
    \item $P_1,\ldots,P_s$ are $s$ distinct places of $F$,
    \item $G$ is a divisor of $F$ such that $Supp(G)\bigcap\{P_1,\ldots,P_s\}=\emptyset$,
\end{itemize}
and for $1\leq i \leq s$ :
\begin{itemize}
    \item $k_i = \deg (P_i)$ is the degree of $P_i,$
    \item $C_i$ is a $[n_i,k_i,d_i]_q$ linear code,
    \item $\pi_i$ is a fixed $\Fq-$linear isomorphism mapping $\mathbb F_{q^{k_i}}$ to $C_i$.
\end{itemize}

Consider the application 
$$
  \alpha~:\begin{array}{lcl}
  \mathcal L(G)  & \longrightarrow & \Fq^n \\
   f & \longmapsto & \left(\pi_1(f(P_1)),\ldots,\pi_s(f(P_s))\right).
\end{array},
$$
where $n=\sum_{i=1}^sn_i$.

\begin{definition}\label{AGgen}
    The image of $\alpha$ is called a generalized algebraic-geometric code, denoted by $C(P_1,\ldots,P_s:G:C_1,\ldots,C_s)$.
\end{definition}

Such a code is well-defined if the application $\alpha$ is injective, that is the case if $\deg(G)<\sum_{i=1}^sk_i$ \cite[Lemma 3.1]{xinggene}. Furthermore, the authors give a lower bound on the minimum distance and the dimension of these codes in the following theorem \cite[Theorem 3.2]{xinggene}. 

\begin{theorem}\label{thmgag}
Under the same notations, if $\deg(G)<\sum_{i=1}^s k_i$ then, the dimension $k$ and the minimum distance $d$ of the code defined by $\alpha$ verify
\begin{itemize}
    \item $k\geq \deg(G)-g+1$, with equality if $\deg(G)\geq 2g-1$,
    \item $d\geq \sum_{i=1}^s d_i -\deg(G) - \max_R\left\{\sum_{i\in R}(d_i-k_i)\right\}$,    
\end{itemize}
where the maximum is extended over all subsets $R$ of ${1,\ldots,s}$ and an empty sum is defined to be 0. 
\end{theorem}

\section{Locality in Generalized AG-codes}

The observation that led to the writing of this paper is the following: if $k_1=\ldots=k_s=:k$, the code defined has locality $k$. More formally,

\begin{proposition}
Let $F$ be an algebraic function fields defined over $\Fq$ of genus $g$.
Let $P_1,\ldots,P_s$ be places of $F$ of degrees $k_i=\deg(P_i)$ respectively, and let $G$ be a divisor such that  $\deg(G)<\sum_{i=1}^sk_i$. For $1 \leq i \leq s$, let $C_i$ be a $[n_i,k_i,d_i]$ $\Fq-$linear code with locality at most $k_i$.
Let $\mathcal{C}=C(P_1,\ldots,P_s:G:C_1,\ldots,C_s)$ be a generalized AG-code as in Definition \ref{AGgen}. If there exists $r\in\mathbb N$ such that for all $1\leq i \leq s$, we have $1<k_i\leq r$ and $n_i>\deg(P_i)$, then $\mathcal{C}$ has locality $r$.
\end{proposition}

\begin{proof}
Let $c=(c_1,\ldots,c_n)\in\mathcal{C}$ be a codeword. For each $1\leq j \leq n$, there exists $1\leq i \leq s$ such that the symbol $c_j$ is a symbol of the linear code $C_i$. Since $C_i$ is a linear code with locality at most $k_i$, one can recover $c_j$ using $k_i\leq r$ other symbols. 

\end{proof}

In order to obtain codes with a given locality $r$, it makes sense to use places $P_1,\ldots,P_s$ of degree $r$, and encode the evaluations at each $P_i$ using the same code $\mathcal C'$. 
We obtain the following. 

\begin{proposition}\label{proppratik}
Let $\mathcal{C}=C(P_1,\ldots,P_s:G:C_1,\ldots,C_s)$ be a generalized AG-code as defined above. Suppose that $\deg P_1=\cdots=\deg P_s=r$ and $\mathcal{C'}=C_1=\cdots=C_s$ is a $[n',r,d']$ linear code with locality $r$. If $2g-1\leq \deg(G) < rs$, then $\mathcal{C}$ is a $[sn',deg(G)-g+1,\geq d'\left(s-\left\lfloor\frac{\deg G}{r}\right\rfloor\right)]$ linear code over $\Fq$ with locality $r$. 
\end{proposition}

\begin{proof}
The code $\mathcal{C}$ is well defined since $\deg(G)<rs=\sum_{i=1}^s\deg P_i$ by Theorem \ref{thmgag}, that also gives the dimension of the code. The length is straightforward. It remains to compute the minimum distance. A function in $f\in\mathcal{L}(G)$ can vanish at most at $\left\lfloor\frac{\deg(G)}{r}\right\rfloor$ places of degree $r$. Moreover, for at least the $s-\left\lfloor\frac{\deg(G)}{r}\right\rfloor$ places where $f$ does not vanish, the projection onto $\mathcal{C}'$ gives at least $d'$ non zero symbols. Consequently, the minimum weight of a codeword in $\mathcal{C}$ is lower bounded by $$d\geq d'\left(s-\left\lfloor\frac{\deg G}{r}\right\rfloor\right).$$
\end{proof}

A specific family of such codes is introduced in Section \ref{asympsection} and its asymptotic properties are studied. For now, let us give several examples of the codes obtained with our new approach. First, we give an example reaching the Singleton bound for LRC. Then, we show the parameters obtained for longer codes. 

\subsection{An optimal example using $\mathbb F_3(x)$}\label{exopt}

Let $F=\mathbb F_3(x)$ be the rational function field over $\mathbb
F_3$. It contains 4 rational places : $P_0$, $P_1$, $P_2$ and $P_\infty$, where $P_i$ can be defined by the polynomial $x-i$ for $0\leq i \leq 2$ and $P_\infty$ is the place at infinity. It also contains three places of degree 2 : $P_1^2$, $P_2^2$ and $P_3^2$, that can be defined by the irreducible polynomials $P_1^2(x)=x^2+2x+2$, $P_2^2(x)=x^2+1$, and $P_3^2(x)=x^2+x+2$ respectively. 
Let $C_1=C_2=C_3=\mathrm{RS}(3,2)=\{(f(0),f(1),f(2))~\mid~f\in\mathbb F_3[x]_{<2}\}$. The code $\mathcal C=C(P_1^2,P_1^2,P_1^2:4P_\infty:\mathrm{RS}(3,2),\mathrm{RS}(3,2),\mathrm{RS}(3,2))$ is a $(9,5)-$ code over $\mathbb F_3$ with locality $2$.

\begin{remark}
This code can be introduced without any geometric notation. In fact, the Riemann-Roch space $\mathcal{L}(4P_\infty)$ is $\mathbb F_3[x]_{<5}$, and for $1\leq i \leq 3$ the evaluations $f(P_i)$ correspond to the projection of $f\in\mathbb F_3[x]_{<5}$ to $\frac{\mathbb F_3[x]}{(P_i^2(x))}\simeq\mathbb F_3[x]_{<2}$.
\end{remark}

Now we explain how to obtain a generating matrix of these code. All the computations were done using Magma \cite{magma}. For clarity, we state this example in terms of polynomials.  

First, we fix the basis of $\mathcal{L}(4P_\infty)$ to be $B=\{1,x,x^2,x^3,x^4\}.$ We compute the evaluation map $\mathcal L(4P_\infty)\longrightarrow F_{P_1^2}\times F_{P_2^2} \times F_{P_3^2}$. In this particular case, it correspond to the remainders of the elements of $B$ modulo the $P_i^2(x)$, for $1\leq i \leq 3$.
We obtain 
$$G_0=\left(
\begin{array}{ccc}
  1 & 1 & 1  \\
  x & x & x  \\
  x+1 & 2 & 2x+1 \\
  2x+1 & 2x & 2x+2 \\
  2 & 1 & 2 
\end{array}
  \right).$$
By fixing the basis of $\mathbb F_3[x]_{<2}$ to be $\{1,x\}$ we can obtain a matrix with coefficients in $\mathbb F_3$:
$$G_1=\left(
\begin{array}{cccccc}
  1 & 0 & 1 & 0 & 1 & 0  \\
  0 & 1 & 0 & 1 & 0 & 1  \\
  1 & 1 & 2 & 0 & 1 & 2 \\
  1 & 2 & 0 & 2 & 2 & 2 \\
  2 & 0 & 1 & 0 & 2 & 0 
\end{array}
  \right).$$
Let $C_i$ denote the $i-$th column of $G_1$. The two first columns $C_1$ and $C_2$ correspond to the evaluation at $P_1$, the columns $C_3$ and $C_4$ to the evaluation at $P_2$ and the last two to the evaluation at $P_3$.
Now consider that the Reed-Solomon code $\mathrm{RS}(3,2)$ defined above is generated, considering the basis $\{1,x\}$ of $\mathbb F_3[x]_{<2}$, by the matrix 
$$G_{\mathrm{RS}}=\left(
\begin{array}{ccc}
  1 & 1 & 1  \\
  0 & 1 & 2  
  \end{array}
  \right).$$

Finally, a generating matrix of $\mathcal C$ is given by 
$$
\left[ [C_1C_2]G_{\mathrm{RS}} \mid [C_3C_4]G_{\mathrm{RS}} \mid [C_5C_6]G_{\mathrm{RS}} \right].
$$
This gives the following matrix over $\mathbb F_3$
$$G=\left(\begin{array}{ccccccccc}
  1 & 1 & 1 & 1 & 1 & 1 & 1 & 1 & 1  \\
  0 & 1 & 2 & 0 & 1 & 2 & 0 & 1 & 2  \\
  1 & 2 & 0 & 2 & 2 & 2 & 1 & 0 & 2 \\
  1 & 0 & 2 & 0 & 2 & 1 & 2 & 1 & 0  \\
  2 & 2 & 2 & 1 & 1 & 1 & 2 & 2 & 2 
\end{array}\right).$$

Moreover, it has locality 2 since if a symbol of a codeword is corrupted, one can reconstruct it using $G_{\mathrm{RS}}$. 

\begin{example}
Let $c=(c_1,\ldots, c_9)\in \mathcal C$, and supppose that the symbol $c_2$ is corrupted, and we want to reconstruct it. We use only the symbols $c_1$ and $c_3$. More precisely, we extract the matrix containing the first and third columns of $G_{\mathrm{RS}}$. It is invertible then one can reconstruct $f_1=f \mod P_1(x)(=f(P_1))$ by applying it to the vector $(c_1,c_3)$. Finally $c_2$ is obtained by evaluating $f_1$ at 1.
\end{example}

According to Proposition \ref{proppratik}, the minimum distance of this code is at least 2. 
Using Magma, we computed that the minimum distance of this code is 3.
Consequently, the code $C(P_1^2,P_1^2,P_1^2:4P_\infty:\mathrm{RS}(3,2),\mathrm{RS}(3,2),\mathrm{RS}(3,2))$ is a [9,5,3] linear code over $\mathbb F_3$ with locality 2, reaching the Singleton-like bound (\ref{SingletonLRC}).

This example generalizes to any prime power $q\geq 3$. and we can similarly define a $[\frac{3}{2}(q^2-q),q^2-q-1,3]_q$ linear code with locality 2, reaching the Singleton bound.  

\begin{proposition}
Let $q\geq3$ be a prime power. Let $\Fq(x)$ be the rational function field over $\Fq$. Let $t={\frac{q^2-q}{2}}$.
Denote by $P_1^2,\ldots,P^2_t$ the $t$ places of degree 2 of $\Fq(x)$. Let $\mathrm{RS}_q(3,2)$ be a RS code of dimension 2 and evaluating at 3 distinct elements of $\Fq$. The code $\mathcal{C}=C(P_1^2,\ldots,P^2_{t}:(q^2-q-2)P_\infty:\mathrm{RS}_q(3,2),\ldots,\mathrm{RS}_q(3,2))$ is a $[\frac{3}{2}(q^2-q),q^2-q-1,3]_q$ linear code with locality 2, reaching the Singleton bound.  
\end{proposition}
\begin{proof}
The length and dimension of the code are straightforward. It remains to compute the minimum distance of $\mathcal{C}$. A function in $f\in\mathcal L((q^2-q-2)P_\infty)$ is a polynomial of degree at most $q^2-q-2$ and thus has at most $q^2-q-2$ zeros. It follows that it can vanish at most at $t-1$ degree 2 places of $\Fq(x)$. If $f$ vanishes at less than $t-1$ places of degree $2$, then the weight of the associated codeword $\mathcal{C}(f)$ is at least 4 since the minimum distance of $\mathrm{RS}_q(3,2)$ is 2. Suppose that $f$ vanishes exactly at $t-1$ degree 2 places. Without loss of generality. Consider that $f$ vanishes at $P_1,\ldots,P_{t-1}$. For all $1\leq i\leq t$, let $p_i$ denote the irreducible polynomial defining the place $P_i$ and let $a_i$ and $a_i^q$ be the roots of $p_i$ in $\mathbb F_{q^2}$. It follows that for all $1\leq i \leq t-1$, the polynomial $p_i$ divides $f$. Thus, $f=c\prod_{i=1}^{t-1}p_i$.
Consider
$$\left(\prod_{i=1}^{t-1}p_i\right)(a_t)=\frac{\prod_{x\in\mathbb F_{q^2}\smallsetminus\{a_t\}}(a_t-x)}{(a_t^q-a_t)\prod_{x\in\mathbb F_{q}}(a_t-x)}.$$
The numerator is the product of all elements of $\mathbb F_{q^2}^*$ and thus is equal to $-1$ by Wilson's Theorem.
Moreover, we have that $\prod_{x\in\mathbb F_{q}}(a_t-x)=a_t^q-a_t$. It follows that
$$\left(\prod_{i=1}^{t-1}p_i\right)(a_t)=\frac{-1}{(a_t-a_t^q)(a_t^q-a_t)}=\frac{1}{(a_t^q-a_t)^2}.$$
Finally, note that $((a_t^q-a_t)^2)^q)=(a_t^q-a_t)^2$, thus is an element of $\Fq$. It results that $\left(\prod_{i=1}^{t-1}p_i\right)(a_t)$ is an element of $\Fq$. Consequently, the polynomial $f=c\left(\prod_{i=1}^{t-1}p_i\right)$ is constant modulo $p_t$, and all of its 3 evaluation using $RS_q(3,2)$ are nonzero, and the weight of $\mathcal C(f)$ is 3. Thus, the minimum distance of $\mathcal{C}$ is 3.
\end{proof}

\subsection{Some longer codes over $\mathbb F_3$}

In the following, we present the results of our experiments over $\mathbb F_3$.
We considered three curves: the projective line over $\mathbb F_3$, the elliptic curve defined by the equation $y^2=x^3+x$ of genus 1, and the Klein quartic defined by $x^4+y^4+1=0$ of genus 3. These last two curves have been chosen because they are maximal over $\mathbb F_9$, and consequently we can expect them to have many places of degree 2 over $\mathbb F_3$. 
More precisely : 
\begin{itemize}
    \item the projective line has 3 places of degree 2,
    \item the elliptic curve has 6 places of degree 2,
    \item and the Klein quartic has 12 places of degree 2.
\end{itemize}

As in the previous example, we evaluate at places of degree 2, and then use the Reed-Solomon code $\mathrm{RS}(3,2)$ to encode the evaluations. Consequently, we obtain codes with locality 2 of length $3s$, where $s$ is the number of degree 2 places used in the construction. As a result, we can build codes with lengths of up to 9 using the projective line, up to 18 using the elliptic curve and up to 36 using the Klein quartic. We constructed these codes using randomly some of the places available for the evaluations. The results are given in Table \ref{data}, where the defect measures the distance from the Singleton bound for LRCs.


\begin{table}[]
    \centering
    \begin{tabular}{cc}
    \begin{tabular}{|c|c|cc|cc|cc|}
    \hline
    & & & \hspace{-2em}$\mathbb F_3(x)$ & & \hspace{-2em}$y^2=x^3+x$ & & \hspace{-2em}$x^4+y^4+1$  \\
    \hline
    $n$ & $k$ & ~~$d$~~ & defect & ~~$d$~~ & defect & ~~$d$~~ & defect \\
    \hline
 & 3 & 4 & 2 & 4 & 2 & 4 & 2  \\ 
9 & 4 & 4 & 1 & 4 & 1 & 4 & 1  \\ 
 & 5 & 3 & 0 & 3 & 0 & 3 & 0  \\ 
 \hline
 & 4 & - & - & 5 & 3 & 6 & 2  \\ 
12 & 5 & - & - & 4 & 2 & 4 & 2  \\ 
 & 6 & - & - & 3 & 2 & 4 & 1  \\ 
 \hline
 & 5 & - & - & 6 & 3 & 6 & 3  \\ 
15 & 6 & - & - & 4 & 4 & 5 & 3  \\ 
 & 7 & - & - & 4 & 2 & 4 & 2  \\ 
 & 8 & - & - & 3 & 2 & 4 & 1  \\   
 \hline
 & 6 & - & - & 6 & 5 & 6 & 5  \\ 
 & 7 & - & - & 6 & 3 & 6 & 3  \\ 
18 & 8 & - & - & 4 & 4 & 4 & 4  \\ 
 & 9 & - & - & 4 & 2 & 4 & 2  \\ 
 & 10 & - & - & 2 & 3 & 3 & 2  \\ 
 \hline
 & 7 & - & - & - & - & 8 & 4  \\ 
 & 8 & - & - & - & - & 6 & 5  \\ 
21 & 9 & - & - & - & - & 5 & 4  \\ 
 & 10 & - & - & - & - & 4 & 4  \\ 
 & 11 & - & - & - & - & 4 & 2  \\ 
 & 12 & - & - & - & - & 4 & 1  \\ 
\hline
 & 8 & - & - & - & - & 8 & 6  \\ 
 & 9 & - & - & - & - & 7 & 5  \\ 
 & 10 & - & - & - & - & 6 & 5  \\ 
24 & 11 & - & - & - & - & 6 & 3  \\ 
 & 12 & - & - & - & - & 4 & 4  \\ 
 & 13 & - & - & - & - & 4 & 2  \\ 
 & 14 & - & - & - & - & 3 & 2  \\
 & 15 & - & - & - & - & 3 & 1  \\
 \hline
 & 9 & - & - & - & - & 8 & 7  \\ 
27 & 10 & - & - & - & - & 8 & 6  \\ 
     & 11 & - & - & - & - & 7 & 5  \\ 
\hline
\end{tabular}
&
\begin{tabular}{|c|c|cc|cc|cc|}
    \hline
    & & & \hspace{-2em}$\mathbb F_3(x)$ & & \hspace{-2em}$y^2=x^3+x$ & & \hspace{-2em}$x^4+y^4+1$  \\
    \hline
    $n$ & $k$ & ~~$d$~~ & defect & ~~$d$~~ & defect & ~~$d$~~ & defect \\
    \hline
 & 12 & - & - & - & - & 6 & 5  \\ 
 & 13 & - & - & - & - & 6 & 3  \\ 
27 & 14 & - & - & - & - & 4 & 4  \\ 
 & 15 & - & - & - & - & 4 & 2  \\ 
 & 16 & - & - & - & - & 3 & 2  \\
\hline
 & 10 & - & - & - & - & 10 & 7  \\ 
 & 11 & - & - & - & - & 8 & 7  \\ 
 & 12 & - & - & - & - & 7 & 7  \\ 
 & 13 & - & - & - & - & 7 & 5  \\ 
30 & 14 & - & - & - & - & 6 & 5  \\ 
 & 15 & - & - & - & - & 6 & 3  \\ 
 & 16 & - & - & - & - & 4 & 4  \\ 
 & 17 & - & - & - & - & 4 & 2  \\ 
 & 18 & - & - & - & - & 3 & 2  \\ 
\hline
 & 11 & - & - & - & - & 10 & 8  \\ 
 & 12 & - & - & - & - & 10 & 7  \\ 
 & 13 & - & - & - & - & 8 & 7  \\ 
 & 14 & - & - & - & - & 8 & 6  \\ 
33 & 15 & - & - & - & - & 6 & 6  \\ 
 & 16 & - & - & - & - & 6 & 5  \\ 
 & 17 & - & - & - & - & 5 & 4  \\ 
 & 18 & - & - & - & - & 4 & 4  \\ 
 & 19 & - & - & - & - & 4 & 2  \\ 
\hline
 & 12 & - & - & - & - & 10 & 10  \\ 
 & 13 & - & - & - & - & 10 & 8  \\ 
 & 14 & - & - & - & - & 8 & 9  \\ 
 & 15 & - & - & - & - & 8 & 7  \\ 
36 & 16 & - & - & - & - & 6 & 8  \\ 
 & 17 & - & - & - & - & 6 & 6  \\ 
 & 18 & - & - & - & - & 5 & 6  \\ 
 & 19 & - & - & - & - & 4 & 5  \\ 
 & 20 & - & - & - & - & 4 & 4  \\ 
\hline
\end{tabular}
    \end{tabular}
    \vspace{0.5em}
    \caption{Parameters of linear codes obtained over $\mathbb F_3$ with locality 2.}
    \label{data}
\end{table}

Table \ref{data} shows that we obtain the same parameters than \cite[Table 2]{chwewexu} when $n=21$ and $k=8$. Moreover, we obtain a $[24,8,8]$ code with locality 2 while their example is of parameters $[24,8,6]$.

\section{Some families of LRCs and asymptotic study}\label{asympsection}

\subsection{Explicit families of LRCs}

Our construction is very close to what can be obtained with concatenated codes. In order to compare both constructions, we introduce a family of concatenated codes and another obtained with our approach. 

\subsubsection{Concatenated Construction.}
This construction is a generalization to an outer AG-code of the construction of \cite[Section VI. A.]{cama} using an outer extended Reed-Solomon code. 

\begin{proposition}[Construction 1]\label{const2}
 Let $F/\mathbb F_{q^r}$ be a function field of genus $g$ containing $s$ rational places, denoted by $P_1,\ldots,P_s$. 
Let $\mathcal C_\mathrm{par}$ the $q-$ary single parity check code of length $r+1$ and dimension $r$, that has minimum distance $2$. For $g-1 < k_0 < s-g+1$, let $G$ be a divisor of $F$ of degree $k_0+g-1$ and $\D=P_1+\cdots+P_s$. Then, the concatenated code $\mathcal C_\mathrm{conc}$ defined by the outer code $C(\D,G)$ and the inner code $\mathcal{C}_\mathrm{par}$ is a $[n,k,\geq d]$ linear code over $\Fq$ with locality $r$, such that
$$n=(r+1)s,$$
$$k=rk_0$$
$$d\geq 2\left(s-\frac{k}{r}-g+1\right).$$
It follows that the rate of this code verifies 
$$\frac{k}{n}\geq \frac{r}{r+1}-\frac{r}{2}\delta - \frac{r(g-1)}{n},$$
where $\delta=\frac{d}{n}.$
\end{proposition}

\begin{proof}
The code $C(\D,G)$ is well defined since $\deg(G)<s$ and its parameters are $[s,k_0,\geq s-\deg(G) =s - k_0 -g +1]$. By Proposition \ref{concprop}, it follows that $\mathcal C_\mathrm{conc}$ is a $[(r+1)s,rk_0,2(s-k_0-g+1)]$ $q-$ary linear code. Then, we obtain
$\frac{k}{r}=k_0\geq s- \frac{d}{2}-g+1,$ and thus
$$\frac{k}{n}\geq\frac{rs}{(r+1)s}-\frac{rd}{2n} - \frac{rg-r}{(r+1)s}=\frac{r}{r+1}\left(1-\frac{r+1}{2}\delta - \frac{g-1}{s}\right).$$
\end{proof}

Note that in this construction, as well as in the known constructions of \cite{taba,tabavl,cama}, the dimension is a multiple of locality. 

\subsubsection{New construction.}
Let us introduce a specific family of codes obtained with our new strategy, using Proposition \ref{proppratik}. 

\begin{proposition}[Construction 2]
 Let $F/\Fq$ be a function field of genus $g$ containing $s$ places of degree $r>1$, denoted by $P_1,\ldots,P_s$. 
Let $\mathcal C_\mathrm{par}$ the $q-$ary single parity check code of length $r+1$ and dimension $r$, that has minimum distance $2$. For $g-1 < k < rs-g+1$, let $G$ be a divisor of $F$ of degree $k+g-1$. Then, the code $C(P_1,\ldots,P_s:G:\mathcal C_\mathrm{par},\ldots,\mathcal C_\mathrm{par})$ is a $[n,k,\geq d]$ linear code over $\Fq$ with locality $r$, such that
$$n=(r+1)s,$$
$$d\geq 2\left(s-\left\lfloor\frac{k+g-1}{r}\right \rfloor\right).$$
It follows that the rate of this code verifies 
$$\frac{k}{n}\geq \frac{r}{r+1}-\frac{r}{2}\delta - \frac{g-1}{n},$$
where $\delta=\frac{d}{n}.$
\end{proposition}

\begin{proof}
By hypothesis on $k$, it follows that $\deg(G)<rs$ and the code is well defined. Moreover, it also ensures that the  expected dimension of the code is obtained, since $\deg(G)>2g-2$. The lower bound $d$ on the minimum distance is directly given by Proposition \ref{proppratik}. The floor $\left\lfloor\frac{k+g-1}{r}\right\rfloor$ can be upper bounded by $\frac{k+g-1}{r}$. It follows that 
$k\geq sr- \frac{rd}{2}-g+1,$ and thus
$$\frac{k}{n}\geq\frac{rs}{(r+1)s}-\frac{rd}{2n} - \frac{g-1}{(r+1)s}=\frac{r}{r+1}\left(1-\frac{r+1}{2}\delta - \frac{g-1}{rs}\right),$$
and the result follows from $n=(r+1)s.$
\end{proof}

\noindent As a result, this construction is of interest for small values of $r$. 

\begin{remark}\label{ifmultiple}
Note that if $g=0$ and $k$ is a multiple of $r$, we get $d\geq 2(s-\frac{k}{r}+1)$.
\end{remark}


\begin{remark}
    Using Constructions 1 and 2 with $g=0$, i.e.~the outer code is a RS code in Construction 1 and we use the rational function field in Construction 2, and if $r$ divides $k$, then the bound on the minimum distance is the same for both construction (Remark \ref{ifmultiple}).
\end{remark}

In the case of a larger fixed genus $g$, and for given dimension $k=k_0r$ and length $n$, suppose that there exists a function field of genus $g$ containing $s$ $\mathbb F_{q^r}-$rational places (for Construction 1), and another containing $s=n/(r+1)$ places of degree $r$ (for Construction 2). Then the codes provided by Construction 2 might have a better minimum distance.
However, this comparison is not very fair. Indeed, the number of places of degree $r$ in a function field defined over $\Fq$ belongs to $\mathcal{O}(\frac{q^r}{r})$, while the number of $\mathbb F_{q^r}-$rational places to $\mathcal{O}(q^r)$. It follows that for a given base field, length and dimension, Construction 1 can be used with function fields of smaller genera than Construction 2. 

In the next sections, we give an asymptotic study of the rates of the codes obtained by both constructions. This gives an idea of the impact of the growth of the genera of the function fields used when the parameters are increasing. 
What remains of this paper is on the use of infinite sequences of function fields reaching this bound in order to obtain families of codes.
The nature of the geometric objects used suggests to consider the asymptotic behaviour of Construction 1 before those of Construction 2. 

\subsection{Asymptotic study of the Concatenated construction}

Before introducing an explicit family, let us give some sufficient conditions for their existence.

\begin{proposition}
Suppose there exists an infinite sequence of function field $\mathcal{F}=(F_1,\ldots,F_\ell,\ldots)$ and $b\in\mathbb R$ such that for all $\ell\geq\ell_0\in\mathbb N$,
$$\frac{B_1(F_\ell/\mathbb F_{q^r})}{g(F_\ell)}\geq b\geq2.$$
Then, Construction 1 provides an infinite family of $q-$ary linear codes with locality $r$ verifying
$$\frac{k}{n}\geq \frac{r}{r+1}\left(1-\frac{r+1}{2}\delta - \frac{1}{b}\right).$$
\end{proposition}
\begin{proof}
    Such an infinite family of codes exists if for all $\ell\geq \ell_0$, one can construct at least one code with Construction 1 while defining the outer code over $F_\ell/\mathbb F_{q^r}$.
    It thus requires there exists an integer $k_0$ such that $g-1<k_0<s-g(F_\ell)+1$ while $F_\ell/\mathbb F_{q^r}$ contains at least $s$ rational places. This is possible if $B_1(F_\ell/\mathbb F_{q^r})-g>g-1$, where $g:=g(F_\ell)$. By hypothesis, we have $B_1(F_\ell/\mathbb F_{q^r})\geq bg$. Consequently, it is sufficient that $(b-2)g>-1$, and that is always the case if $b\geq2$. The ratio of the codes directy comes from Proposition \ref{const2}.
\end{proof}

The existence of infinite family of codes defined by this construction is ensured by sequences of function fields reaching the Drinfeld-Vladut bound, such as the recursive tower of function fields defined by Garcia and Stichtenoth \cite{gast}. This tower $\mathcal{T}=(T_1/\mathbb F_{q^2},\ldots,T_\ell\mathbb F_{q^2},\ldots)$ is defined as follows. 

\begin{definition}\label{GS}
Set $T_1=\mathbb F_{q^2}(x_1)$ the rational function field over $\mathbb F_{q^2}$, and for $i\geq 1$ we define 
$$T_{i+1}=T_i(z_{i+1}),$$
where $z_{i+1}$ satisfies the equation 
$$z_{i+1}^q+z_{i+1}=x_i^{q+1}\text{ with }x_i=\frac{z_i}{x_{i-1}}\text{ for }i\geq2.$$    
\end{definition}

By \cite[Theorem 2.10 and Proposition 3.1]{gast}, we have the following
\begin{proposition}
    Let $\ell\geq 3$, then the genus $g(T_\ell)$ and the number of rational place $B_1(T_\ell)$ of $T_\ell$ verify
    $$g(T_\ell)\leq q^\ell + q^{\ell-1}\text{ and }B_1(T_\ell/\mathbb F_{q^2})\geq (q^2-1)q^{\ell-1}+2q.$$
\end{proposition}
This gives us directly the useful following result. 
\begin{corollary}\label{lemmeC1}
    Let $\ell\geq 3$, then
    $$\frac{B_1(T_\ell/\mathbb F_{q^2})}{g(T_\ell)}\geq q-1.$$
\end{corollary}
It follows that the tower $\mathcal T/\mathbb F_q^2$ is asymptotically optimal, i.e. reaches the Drinfeld-Vladut Bound. This is an ideal frame for applying Construction 1.

\begin{proposition}\label{coroC1}
    Let $q$ be a prime power and $r$ an even integer, except $q=r=2$. Then, Construction 1 provides an infinite family of linear code with locality $r$ verifying
$$\frac{k}{n}\geq \frac{r}{r+1}\left(1-\frac{r+1}{2}\delta - \frac{1}{q^{\frac{r}{2}}-1}\right).$$
\end{proposition}
\begin{proof}
    If $q^r$ is a square, then the recursive tower $\mathcal{T}=(T_1,\ldots)$ over $\mathbb F_{q^r}$ of Garcia-Stichtenoth reaches the Drinfeld-Vladut bound. In particular, for all $\ell\geq3$, we have $\frac{B_1(T_\ell)}{g(T_\ell)}\geq q^{\frac{r}{2}}-1$ (Corollary \ref{lemmeC1}), and this value is at least 2 except when $q=r=2$. 
\end{proof}

\begin{example}\label{exC1}
    For $q=4$ and $r=2$, we obtain $\frac{k}{n}\geq \frac{1}{3}-\delta$.
\end{example}

\noindent Note that for $r=2$, this gives the same bound as in \cite[Theorem 3.7]{xingalllast}.

\subsection{Asymptotic study of the new construction}

First, we give sufficient conditions for applying the construction. 

\begin{proposition}\label{asympC1}
Suppose there exists an infinite sequence of function field $\mathcal{F}=(F_1,\ldots,F_\ell,\ldots)$ and $b_r\in\mathbb R$ such that for all $\ell\geq\ell_0\in\mathbb N$,
$$\frac{B_r(F_\ell/\mathbb F_{q})}{g(F_\ell)}\geq rb_r\geq2.$$
Then, Construction 2 provides an infinite family of $q-$ary linear codes with locality $r$ verifying
$$\frac{k}{n}\geq \frac{r}{r+1}\left(1-\frac{r+1}{2}\delta - \frac{1}{rb_r}\right).$$
\end{proposition}
\begin{proof}
    Similarly to the case of Construction 1, such an infinite family of codes exists if for all $\ell\geq \ell_0$, one can construct at least one code with Construction 1 using $F_\ell/\mathbb F_{q}$. Let $g$ denotes $g(F_\ell)$.
    It thus requires there exists an integer $k$ such that $g-1<k<s-g+1$ while $F_\ell/\mathbb F_{q}$ contains at least $s$ places of degree $r$. This is possible if $rB_r(F_\ell/\mathbb F_{q})-g>g-1$. By hypothesis, we have $B_r(F_\ell/\mathbb F_{q})\geq b_rg$. Consequently, it is sufficient that $(rb_r-2)g>-1$, and that is always the case if $rb_r\geq2$. The ratio of the codes directly comes from Proposition \ref{const2}.
\end{proof}

In \cite{baro4}, Ballet and Rolland studied the descent of the tower $\mathcal{T}/\mathbb F_{q^2}$ (Definition \ref{GS}) to the field of constant $\Fq$. More precisely,

\begin{definition}\label{towerfq}
    Let $\mathcal{U}/\mathbb F_q=(U_1/\Fq,\ldots,U_\ell/\Fq,\ldots)$ be the tower of function fields such that $T_i=\mathbb F_{q^2}U_i$, i.e. $T_i/\mathbb F_{q^2}$ is the constant field extension of $U_i/\Fq$.
\end{definition}
The authors also proved that these towers reach the Drinfeld-Vladut bound at order 2 \cite[Proposition 3.3]{baro4}. This allows us to prove the existence of infinite families of linear code with locality 2 thanks to Construction 2.

\begin{corollary}
    Let $q>3$ be a prime power. Then, Construction 1 provides an infinite family of linear code with locality $2$ verifying
$$\frac{k}{n}\geq \frac{2}{3}\left(1 - \frac{q}{q^2-q-2}\right)-\delta.$$
\end{corollary}
\begin{proof}
    We consider the tower $\mathcal{U}/\Fq$ (Definition \ref{towerfq}). By \cite[Proposition 3.3]{baro4}, it is proven that for $\ell\in\mathbb N $, 
    $$B_1(U_\ell/\Fq)+2B_2(U_\ell/\Fq)\geq q^2-1\geq (q^2-1)q^{\ell-1}\text{ and }g(U_\ell)\leq q^\ell+q^{\ell-1}.$$
     Moreover, it is also proven \cite[Proof of Lemma 3.1]{baro4} that $B_1(U_\ell/\Fq)\leq2q^2.$ It follows that $2B_2(U_\ell/\Fq)\geq \frac{1}{2}\left((q^2-1)q^{\ell-1}-2q^2\right)$. Then we obtain, 
    $$\frac{B_2(U_\ell/\Fq)}{g(U\ell)}\geq \frac{q^{\ell+1}-q^{\ell-1}-2q^2}{2(q^\ell+q^{\ell-1})}=\frac{q-1}{2}-\frac{q^2}{q^\ell+q^{\ell-1}}.$$
    In particular, for $\ell \geq 3$, $$\frac{B_2(U_\ell/\Fq)}{g(U\ell)}\geq\frac{q-1}{2}-\frac{1}{q}.$$
    According to the Proposition \ref{asympC1}, it is required that $2(\frac{q-1}{2}-\frac{1}{q})\geq 2$. That is the case if $q^2-3q-2\geq 0$ and thus if $q>3$.
\end{proof}

\begin{remark}
    While $\ell$ is increasing, the rate of the codes defined using $U_\ell$ tends to verify $\frac{k}{n}\geq \frac{2}{3}(1-\frac{3}{2}\delta-\frac{1}{q-1})$. This is exactly the bound obtained for Construction 1 in Proposition \ref{coroC1}, specialized to locality 2. More generally, according to the Drinfeld-Vladut Bound, the best rate that can be obtained is
$$\frac{n}{k}\geq\frac{r}{r+1}\left(1-\frac{r+1}{2}\delta-\frac{1}{q^{\frac{r}{2}}-1}\right).$$
\end{remark}

\begin{remark}
Construction 2 requires asymptotically a large number of places of degree $r$. Such objects can be obtained by the descent to $\Fq$ of function fields defined of $\mathbb F_{q^r}$ reaching the Drinfeld-Vladut bound (\cite{cacrxiya}, see \cite{survey}). The sequences studied in \cite{baro4} for $r=2$ or $q=2$ and $r=4$ are convenient for our study.
\end{remark}

\begin{remark}
    One can construct directly codes with locality 2 of odd dimension with the new Construction 2, while it is not possible with Construction 1. In the literature, it is classical to obtain LRCs whose dimension is a multiple of the locality, then some techniques can be used to obtain different dimension, as in \cite{xingalllast}. 
\end{remark}

\begin{remark}
    A reasonable question is whether one can build codes with locality 2 of any dimension $k\in\mathbb N$ thanks to Construction 2. 
    It is not possible with the tower used previously. But it might be possible using the densified version of the tower $\mathcal{U}$ introduced by Ballet and Rolland in \cite{baro4}.
\end{remark}

\begin{remark}
    Although it was quite natural to consider only places of a fixed degree $r$, one can extend our new construction to the use of places of smaller degree, provided that we combine them to obtain spaces of dimension $r$. Moreover, one can also consider generalized evaluation maps, and for instance use the local expansion at order $r$ at rational places. 
\end{remark}

\begin{remark}
    With additional constraints, one can define an evaluation and interpolation multiplication algorithm following the method of Chudnovsky and Chudnovsky \cite{chch,survey}, where the evaluation map defines a generalized AG-code. Then, the minimum distance of the code would be equal to its dimension by \cite{shtsvl}. It could be an interesting work to see whether this argument can give a better bound on the minimum distance for some codes. 
\end{remark}

\begin{remark}
 The construction introduced in this document might be generalized in order to obtain codes with hierarchical locality \cite{hierar}.
\end{remark}

\section{Aknowledgments}

The author is deeply grateful to the French ANR BARRACUDA (ANR-21-CE39-0009-BARRACUDA) for its support, and to several of its members for the many inspiring and helpful discussions.

\bibliographystyle{plain}
\bibliography{biblio}

\end{document}